\documentclass[openacc]{rstransa}

\newtheorem{definition}{\bf Definition}
\newtheorem{theorem}{\bf Theorem}
\newtheorem{lemma}{\bf Lemma}
\newtheorem{assumption}{\bf Assumption}

\newcommand{\indep}{\perp \!\!\! \perp}

\usepackage{amsmath}

\usepackage{bm}

\usepackage{color}

\titlehead{Research}

\usepackage[T1]{fontenc}			
\usepackage[sc,osf]{mathpazo}   	

\begin{document}

\title{Comparing the cost of violating causal assumptions in Bell experiments: locality, free choice and arrow-of-time}

\author{
Pawel Blasiak$^{\scriptscriptstyle 1,2}$ and Christoph Gallus$^{\scriptscriptstyle 3}$}

\address{$^{\scriptscriptstyle 1\,}$Institute for Quantum Studies, Chapman University, Orange, CA 92866, USA\\
$^{\scriptscriptstyle 2\,}$Institute of Nuclear Physics, Polish Academy of Sciences, 31342 Kraków, Poland\\
$^{\scriptscriptstyle 3\,}$Technische Hochschule Mittelhessen, 35390 Gießen, Germany}

\subject{quantum physics}


\corres{Pawel Blasiak\\
\email{pawel.blasiak@ifj.edu.pl}}

\begin{abstract}
The causal modelling of Bell experiments relies on three fundamental assumptions: locality, freedom of choice, and arrow-of-time. It turns out that nature violates Bell inequalities, which entails the failure of at least one of those assumptions. Since rejecting any of them -- even partially -- proves to be enough to explain the observed correlations, it is natural to ask about the cost in each case. This paper follows up on the results in PNAS {\bf 118} e2020569118 (2021), showing the equivalence between the locality and free choice assumptions, adding to the picture retro-causal models explaining the observed correlations. Here, we consider more challenging causal scenarios which allow only single-arrow type violations of a given assumption. The figure of merit chosen for the comparison of the causal cost is defined as the minimal frequency of violation of the respective assumption required for a simulation of the observed experimental statistics.
\end{abstract}


\begin{fmtext}

\end{fmtext}


\maketitle

\section{Introduction}

Violation of Bell inequalities in quantum theory~\cite{Be93,BrCaPiScWe14,Sc19} --- confirmed in a number of ingenious experiments~\cite{AsDaRo82,HeBeDrReKaBlRu15,GiVeWeHaHoPhSt15,ShMeChBiWaStGe15,As15,RaHaHoGaFrLeLi18,BIGBellCollaboration18} --- is a subject of heated debates in the quantum foundations community~\cite{No11a,Wi14a,JPA14}. This is because of the challenge it poses to the common classical worldview. Formally, Bell's theorem is a consequence of three causal assumptions that are neatly formulated in the language of \textit{causal Bayesian models}~\cite{Pe09,SpGlSc00,PeGlJe16}. The latter can be seen as the modern formalisation of the "classical" concept of causality which subsumes the Bell's original ideas~\cite{Be93,Va82}. Once this concept of causality (and its mathematical framework) is accepted as a fundamental principle, the assumptions of \textit{locality} and \textit{arrow-of-time} are the expression of our understanding of how variables are organised in space-and-time, while \textit{free choice} (or \textit{measurement independence}) is an assumption about the role of external parameters in modelling an experiment. When taken together, those three assumptions are contradicted by the observed quantum-mechanical statistics obtained in the so-called Bell experiment. It means that, if "classical" causality is to be maintained, at least one of those assumptions must fail. Needless to say, the rejection of either of them has profound philosophical and interpretational consequences~\cite{No17,La19,Ma19}, with research programs on each side. Violation of locality would require some physical mechanism operating instantaneously, irrespective of the distance and broadcasting to every corner of the universe, which stays in tension with Einstein's theory of relativity~\cite{Be93}. Retrocausality makes problematic the distinction between causes and effects, which are conventionally differentiated by the arrow-of-time~\cite{Pr96,WhAr20,LePu17,Ad22}. The rejection of free choice provokes a rethinking of the concept of observers and inflicts conspiratorial elements into the experimental analysis (or even super-determinism in the extreme case)~\cite{Br88,Ho16a,HoPa20}. Finally, one may also try to dismiss the fundamental role of causality in physics~\cite{No03}, modify its meaning~\cite{CaLa14a,WiCa17,AlBaHoLeSp17,Dz22a,CaLe23,Ad23}, or just remain content with purely operational accounts of physical theories~\cite{BuGrLa95,Li37}. In the following, we stick to the standard or "classical" approach to causality as laid out by Pearl and others~\cite{Pe09,SpGlSc00,PeGlJe16}, and henceforth drop the qualifier "classical".

In this paper, we take an impartial position with a view to analysing the consequences of a violation of each of the three causal assumptions. Within the causal Bayesian model framework, we first discuss a baseline Bell scenario which obeys all three assumptions. It allows us to see a violation of a given assumption as a departure from the baseline case by adding certain arrows. Clearly, the reasoning can be reversed, i.e., admitting the presence of a given arrow in a causal account of the observed statistics indicates a violation of the corresponding assumption. This observation provides a unified means for a quantitative comparison of all three assumptions in the Bell scenario using the concept of a \textit{causal fraction}. It aims to measure the cost/weight/strength of a given arrow by answering a simple question: \textit{"How often does a given arrow need to have been actually at work to account for the observed experimental statistics?"} We prove the equivalence of all three assumptions --- locality, arrow-of-time and free choice --- based on equating the term \textit{"cost"} with the causal fraction measure. In this sense, the causal framework in itself does not give preference to explanations based on the violation of either of the assumptions. This completes the results in~\cite{BlPoYeGaBo21}, which now encompass all three causal assumptions, and makes explicit their scope of validity within the causal Bayesian model framework.

\section{Causal modelling of Bell experiment}

Let us consider the usual Bell-type scenario with two parties, called Alice and Bob, performing experiments on two separated systems~\cite{Be93,BrCaPiScWe14,Sc19}. We will assume that the \textit{measurement settings} are labelled by $x\in\mathcal{X}$ for Alice and $y\in\mathcal{Y}$ for Bob, and their \textit{measurement outcomes} are labelled respectively $a\in\mathcal{A}$ and $b\in\mathcal{B}$. 
For the sake of generality, we do not restrict the cardinality of the sets $\mathcal{A}$, $\mathcal{B}$, $\mathcal{X}$ and $\mathcal{Y}$ (just assume that they are finite to avoid technical issues). A Bell experiment consists of a series of trials in which Alice and Bob each choose a setting and make a measurement registering the outcomes. Then, after many repetitions, they compare their results, arriving at statistics given by the joint distribution $P_{abxy}$. It is conventionally split into the conditional probability $P_{ab|xy}$ of Alice and Bob observing outcomes $a$ and $b$ given the choice of settings $x$ and $y$, and the distribution of settings  $P_{xy}$ chosen in the experiment. This follows the idea of understanding a measurement as probing an existing system by questions $x$, $y$, which are chosen by the experimenters Alice and Bob (or some devices on their behalf), to which the system responds with the respective answers $a$ and $b$. In the following, we will call conditional probabilities $P_{ab|xy}$ the  \textit{behaviour}, while $\big\{P_{ab|xy}\,,P_{xy}\big\}$ will be called the \textit{(experimental) statistics}.
Note that without assuming anything about the causal structure underlying the experiment, any experimental statistics is admissible as long it is well-normalised, i.e. we have $\sum_{a,b}\,P_{ab|xy}\,=\,1$ for each $(x,y)\in\mathcal{X}\times\mathcal{Y}$, and $\sum_{x,y}\,P_{xy}\,=\,1$.

Quantum theory gives a method for calculating the behaviour $P_{ab|xy}$ based on the formalism of Hilbert spaces. Strikingly, it raises a challenge to our most cherished intuitions/assumptions about the workings of the world. The first assumption concerns \textit{locality}. A Bell experiment can be arranged in a way that Alice and Bob are space-like separated, which suggests the lack of any influence between Alice's and Bob's sides. The second assumption relates to \textit{arrow-of-time}. It stipulates that the causes must precede the effects, and hence nothing can affect the past, rendering \textit{retrocausal} explanations implausible. The third assumption is about \textit{freedom of choice} (or \textit{measurement independence}). It amounts to the belief that experimental settings can be chosen independently of anything related to the investigated systems, and this extends to the choices made by random generators (on Alice's and Bob's behalf). Those three assumptions underly the famous Bell inequalities, whose violation by quantum mechanical correlations (observed in numerous experiments~\cite{AsDaRo82,HeBeDrReKaBlRu15,GiVeWeHaHoPhSt15,ShMeChBiWaStGe15,As15,RaHaHoGaFrLeLi18,BIGBellCollaboration18}) indicates the necessity of revising some fundamental concepts in our approach to modelling physical reality. 

In the following, we take the causal point of view to modelling the Bell experiment. We will use causal DAGs (\textit{directed acyclic graphs}) to encode the various assumptions as well as their violations (in terms of additional arrows) to discuss possible ways of explaining the observed statistics. The tool of choice in this analysis is the framework of Bayesian causal models~\cite{Pe09,SpGlSc00,PeGlJe16}, which can be seen as a modern formalisation of Bell's original ideas~\cite{Be93}.

\subsubsection*{($\bm{*}$) {Baseline scenario}}
\begin{figure}[h]
\centering
\includegraphics[width=0.85\columnwidth]{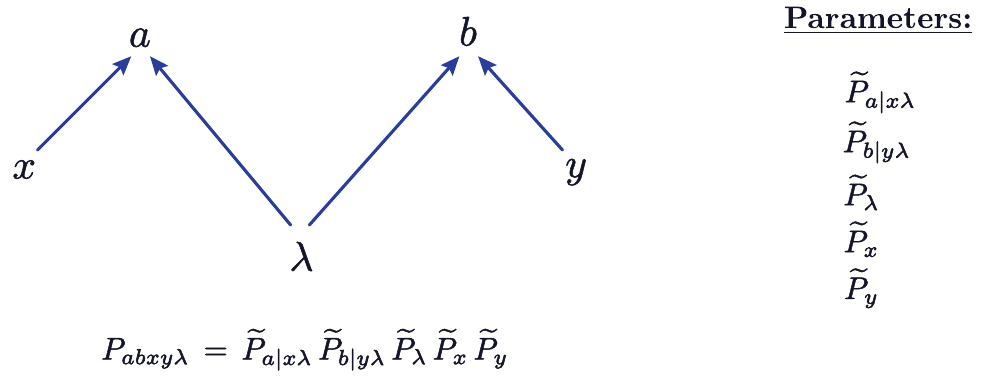}
\caption{\label{Fig-Bell-Baseline}\mbox{{\bf\textsf{Baseline Bell scenario.}}} This causal DAG encodes all three assumptions: \textit{locality}, \textit{arrow-of-time} and \textit{free choice}. Upon specifying the five distributions as parameters in the model (shown on the right-hand side) one obtains the joint probability $P_{abxy\lambda}$ as generated by the product formula (shown at the bottom). This baseline scenario fails to explain the observed quantum correlations which violate Bell inequalities.\vspace{-0.2cm}
}
\end{figure}

As explained above, there are four observed variables which are the choice of settings $x$,$y$ and the measurement outcomes $a$,$b$. In addition, we admit in the model some unobserved latent (or hidden) variable $\lambda$. Then, a general causal DAG, as depicted in Fig.~\ref{Fig-Bell-Baseline}, can serve as a causal model of the Bell experiment. The DAG encodes the three assumptions \textit{locality}, \textit{arrow-of-time} and \textit{free choice} in the following manner.

In the diagram, the variables $x$, $a$ pertaining to Alice are located on the left, while the variables $y$, $b$ pertaining to Bob are on the right. The \textit{locality} assumption prohibits any arrow from Alice's to Bob's side and vice versa. The time in the diagram goes from bottom to top, which is reflected in the placement of the variables. Thus the hidden variable $\lambda$ is interpreted as a common cause in the past, which may be correlated with future events. The \textit{arrow-of-time} assumption requires that all the arrows, representing causal relationships between the variables, follow the direction of time. A variable without incoming arrows is called a free (or exogenous) variable. It is supposed to be supplied from outside and hence unaffected by any factor described in the model. The \textit{free choice} assumption stipulates the choice of settings $x$ and $y$ to be free variables. Those three assumptions rule out certain arrows in the diagram leading to the causal DAG in Fig.~\ref{Fig-Bell-Baseline}.

In order to complete the model, the DAG needs to be furnished with the set of parameters denoted throughout the paper with tildes $\widetilde{P}$. Then, via the product decomposition (which follows from the Markov condition), one obtains the joint distribution of all variables $P_{abxy\lambda}$ in the model. This allows calculating the experimental statistics (as a simple application of Bayes' rule)
\begin{eqnarray}\label{base-1}
P_{ab|xy}&=&\sum_\lambda\ \widetilde{P}_{a|x
\lambda}\,\widetilde{P}_{b|y\lambda}\,\widetilde{P}_\lambda\,,\\\label{base-2}
P_{xy}&=&\widetilde{P}_{x}\, \widetilde{P}_{y}\,.
\end{eqnarray}
The experimental statistics $\big\{P_{ab|xy}\,,P_{xy}\big\}$ is explainable by the causal structure in Fig.~\ref{Fig-Bell-Baseline} only if it can be obtained in this way. Notably, the causal structure imposes constraints on the possible correlations in any compatible distribution. Firstly, the causal DAG induces certain conditional independencies as neatly described by the so-called \textit{d-separation} rules~\cite{Pe09}. They are fulfilled in the quantum mechanical statistics, where they boil down to non-signalling conditions.\footnote{\label{non-sig}Non-signalling boils down to two constraints on conditional distributions: $P_{b|xy}=\sum_a P_{ab|xy}=\sum_a P_{ab|x'y}=P_{b|x'y}$ (Alice cannot signal to Bob) and $P_{a|xy}=\sum_b P_{ab|xy}=\sum_b P_{ab|xy'}=P_{a|xy'}$ (Bob cannot signal to Alice) which holds for all $x,x'\in\mathcal{X}$ and $y,y'\in\mathcal{Y}$.\vspace{0.2cm}} Secondly, the behaviour compatible with the DAG in Fig.~\ref{Fig-Bell-Baseline} needs to satisfy the famous \textit{Bell inequalities}~\cite{Be93,BrCaPiScWe14,Sc19}. Quantum theory does not satisfy those inequalities, and hence the DAG in Fig.~\ref{Fig-Bell-Baseline} is not the right causal structure capable of explaining quantum mechanical correlations. It means that in order to explain the observed correlations (and remain within the conventional causal framework~\cite{Pe09,SpGlSc00,PeGlJe16}), the baseline scenario needs to be modified by adding or inverting some of the arrows. This, of course, entails a violation of at least one of the considered assumptions. Below we discuss three generic types of modifications to the baseline causal DAG that enable explaining the observed statistics. Here, we insist on minimal alterations to the DAG in Fig.~\ref{Fig-Bell-Baseline}, which consists in adding only a single arrow. We are guided by the universal principle of graphical causal models saying that if a single arrow is able to explain the statistics, then a richer graph will do it as well. Such minimal modifications will turn out to be sufficient and thus render considering more arrows redundant (or less interesting).

\subsubsection*{(\textsf{NL}) {Violation of locality}}

\begin{figure}[h]
\centering
\includegraphics[width=0.85\columnwidth]{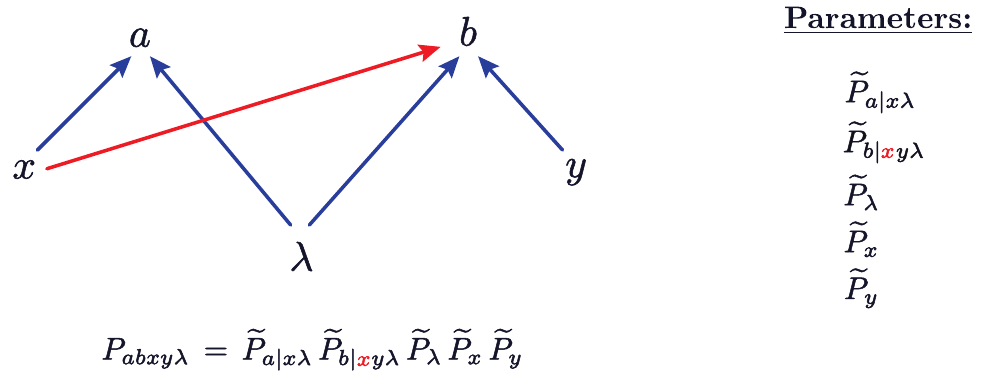}
\caption{\label{Fig-Bell-NL}\mbox{{\bf\textsf{Violation of locality.}}} Alice's choice of setting $x$ can influence Bob's outcome $b$. This stipulates the existence of some physical mechanism $x\rightarrow b$ (depicted by the red arrow) allowing instantaneous propagation of the causes to Bob's side, whatever the distance and location. Note that it requires broadcasting to every place in the universe, as Bob's location may be unknown.\vspace{-0.2cm}
}
\end{figure}
%

%
Relaxing the \textit{locality} assumption makes it possible for Alice's side to influence Bob's side and vice versa. We consider only the direction from Alice to Bob (since the problem is symmetrical). The only possibility is the arrow $x\rightarrow b$ in Fig.~\ref{Fig-Bell-NL}. A purported influence that $x$ might have on $b$ is called \textit{parameter dependence}. All other options need to be rejected. Note that adding the arrows $x\rightarrow y$ or $a\rightarrow y$ would violate the free choice assumption. As for inverting the arrow $\lambda\rightarrow a$ or $x\rightarrow a$, it would introduce retrocausality (as well as violate free choice in the latter case). The remaining possibility is to introduce the arrow $a\rightarrow b$, called \textit{outcome dependence}. However, it turns out to be insufficient to reproduce the quantum mechanical statistics in the Bell scenario with more than two settings.\footnote{\label{footnote-outcome-dependence}For a two-outcome, two-setting scenario, it is possible to simulate the statistics using only \textit{outcome dependence} (i.e. by introducing arrow $a\rightarrow b$ in Fig.~\ref{Fig-Bell-Baseline}) for an even broader class of non-signalling distributions (i.e., simulate the so-called PR-boxes~\cite{PoRo94}). However, if either Alice or Bob has a choice of three settings, then one can derive Bell-like inequalities for the DAG with \textit{outcome dependence} which are violated by quantum mechanical correlations (see~\cite{ChKuBrGr15} for a three-outcome, three-setting scenario).}

The causal DAG in Fig.~\ref{Fig-Bell-NL} entails certain constraints for the observed experimental statistics, which has to be in the form
\begin{eqnarray}\label{NL-1}
P_{ab|xy}&=&\sum_\lambda\ \widetilde{P}_{a|x
\lambda}\,\widetilde{P}_{b|xy\lambda}\,\widetilde{P}_\lambda\,,\\\label{NL-2}
P_{xy}&=&\widetilde{P}_{x}\, \widetilde{P}_{y}\,,
\end{eqnarray}
where, again, $\widetilde{P}$\,'s are parameters of the model that need to be specified.


\subsubsection*{(\textsf{R}) {Violation of arrow-of-time (retrocausality)}}

\begin{figure}[h]
\centering
\includegraphics[width=0.85\columnwidth]{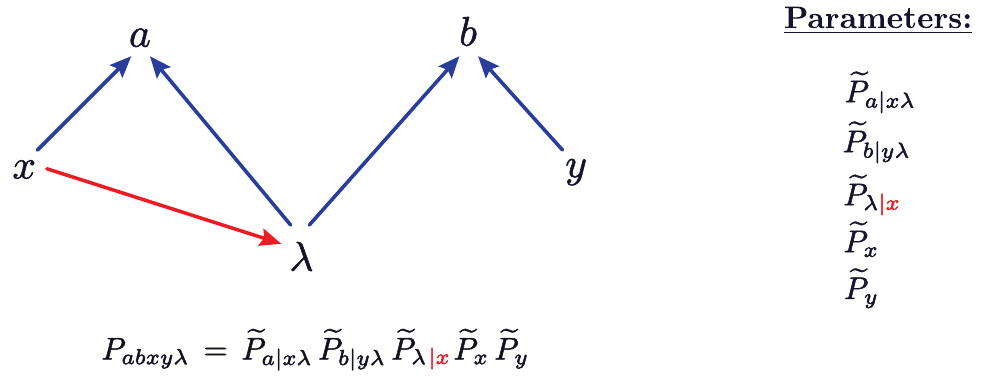}
\caption{\label{Fig-Bell-R}\mbox{{\bf\textsf{Violation of arrow-of-time (retrocausality).}}}  Alice's choice of setting $x$ affects hidden variable $\lambda$ in the past. It means rejection of the role of time in the relation between causes and effects. This allows a mechanism to build a 'local channel' between Alice and Bob $x\rightarrow\lambda\rightarrow b$, which goes 'through the past'.\vspace{-0.2cm}
}
\end{figure}

Let us admit \textit{retrocausal} actions in the form of the arrow $x\rightarrow \lambda$, which leads to Fig.~\ref{Fig-Bell-R}. It means that Alice's choice can influence the hidden variable $\lambda$ in the past (the case of Bob affecting $\lambda$ in the past is analogous). If locality and free choice are preserved, all other options need to be rejected. Note that inverting the arrow $x\rightarrow a$ would compromise Alice's free choice in addition to being retrocausal. On the other hand, the inversion of the arrow $\lambda\rightarrow a$ is tantamount to a non-local mechanism of the outcome-dependence type (via $\lambda$), which is insufficient to reproduce quantum correlations; cf. Footnote~\ref{footnote-outcome-dependence}. Thus we are left with the arrow $x\rightarrow \lambda$ as the only interesting retrocausal option which respects the other two assumptions.

Upon specifying parameters $\widetilde{P}$ in the model, one obtains the joint distribution $P_{abxy\lambda}$ in the product form shown at the bottom of Fig.~\ref{Fig-Bell-R}, and then calculates the  experimental statistics
\begin{eqnarray}\label{R-1}
P_{ab|xy}&=&\sum_\lambda\ \widetilde{P}_{a|x
\lambda}\,\widetilde{P}_{b|y\lambda}\,\widetilde{P}_{\lambda|x}\,,\\\label{R-2}
P_{xy}&=&\widetilde{P}_{x}\, \widetilde{P}_{y}\,,
\end{eqnarray}
which is the general form allowed by the causal DAG in Fig.~\ref{Fig-Bell-R}.

\subsubsection*{(\textsf{NF}) {Violation of free choice (measurement dependence)}}

\begin{figure}[h]
\centering
\includegraphics[width=0.85\columnwidth]{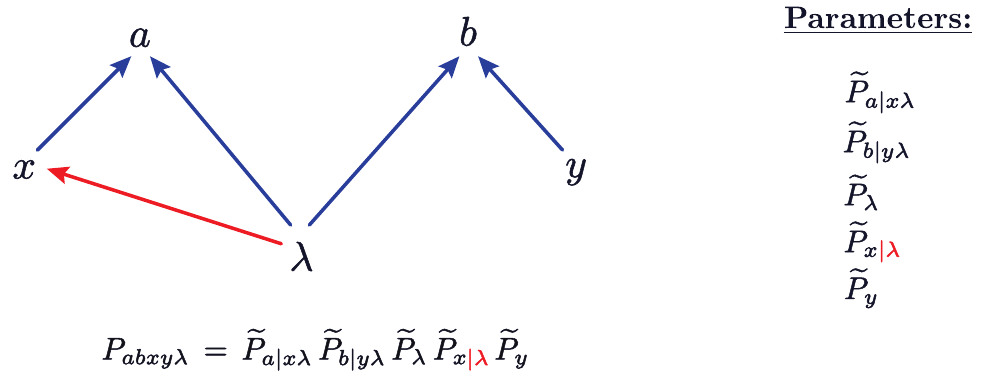}
\caption{\label{Fig-Bell-NF}\mbox{{\bf\textsf{Violation of free choice (measurement dependence).}}}  Alice's choice of setting $x$ is dictated by the hidden variable $\lambda$. In this case, the arrow $\lambda\rightarrow x$ compromises Alice's free choice. It is the modeller who has the freedom to set the variable $\lambda$, so as to get the desired distribution of settings on Alice's side. Note that in the extreme case, this leads to super-determinism or some milder version thereof.\vspace{-0.2cm}
}
\end{figure}

As the last option, we consider the possibility of violating Alice's \textit{free choice} while keeping locality and avoiding retrocausality. Here, Bob's freedom is retained. The variable $x$ is no longer free, which allows it to accept incoming arrows. There is only one such possibility that respects the other assumptions, that is adding the arrow $\lambda\rightarrow x$. Clearly, inverting the arrow $x\rightarrow a$ contradicts the arrow-of-time, and any arrow from Bob's side is ruled out by locality.

Let us note that in this case, the variable $x$ is \textit{not} exogenous anymore. It is a direct consequence of rejecting Alice's free choice. The parameter $\widetilde{P}_{x|\lambda}$ depends on the hidden variable $\lambda$. However, it still allows reproducing any desired distribution ${P}_x$ by exploiting the fact that the variable $\lambda$ is free (exogenous in the model), i.e. its choice is at the discretion of the modeller. The observed experimental statistics results from specifying the parameters $\widetilde{P}$ in the causal DAG in Fig.~\ref{Fig-Bell-NF}, which leads to the following expression
\begin{eqnarray}\label{NF-1}
P_{ab|xy}&=&\sum_\lambda\ \widetilde{P}_{a|x\lambda}\,\widetilde{P}_{b|y\lambda}\,\widetilde{P}_{\lambda|x}\,,\\\label{NF-2}
P_{xy}&=&\sum_\lambda\ \widetilde{P}_{x|\lambda}\,\widetilde{P}_\lambda\, \widetilde{P}_{y}\,.
\end{eqnarray}
As usual, this neat form is obtained from the joint distribution $P_{abxy\lambda}$ via Bayes' rule.\footnote{In this case, let us make this calculation explicitly:
\begin{eqnarray}\nonumber
&&\hspace{-0.5cm}P_{ab|xy}\,=\,
\frac{{\sum}_\lambda\,P_{abxy\lambda}}{P_{xy}}\,\stackrel{\scriptscriptstyle Fig.\ref{Fig-Bell-NF}}{=}\,
\frac{{\sum}_\lambda\,\widetilde{P}_{a|x\lambda}\widetilde{P}_{b|y\lambda}\widetilde{P}_{\lambda}\widetilde{P}_{x|\lambda}\widetilde{P}_{y}}{{\sum}_{ab\lambda}\,\widetilde{P}_{a|x\lambda}\widetilde{P}_{b|y\lambda}\widetilde{P}_{\lambda}\widetilde{P}_{x|\lambda}\widetilde{P}_{y}}\,=\,
\frac{{\sum}_\lambda\,\widetilde{P}_{a|x\lambda}\widetilde{P}_{b|y\lambda}\widetilde{P}_{x|\lambda}\widetilde{P}_{\lambda}\widetilde{P}_{y}}{{\sum}_{\lambda}\,\widetilde{P}_{x|\lambda}\widetilde{P}_{\lambda}\widetilde{P}_{y}}\,=\,\frac{{\sum}_\lambda\,\widetilde{P}_{a|x\lambda}\widetilde{P}_{b|y\lambda}\widetilde{P}_{x|\lambda}\widetilde{P}_{\lambda}\widetilde{P}_{y}}{\widetilde{P}_{x}\widetilde{P}_{y}}\,=\,{\sum}_\lambda\,\widetilde{P}_{a|x\lambda}\widetilde{P}_{b|y\lambda}\widetilde{P}_{\lambda|x}\,,
\\\nonumber
&&\hspace{-0.27cm}P_{xy}\,=\,
{\sum}_{ab\lambda}\,P_{abxy\lambda}\,\stackrel{\scriptscriptstyle Fig.\ref{Fig-Bell-NF}}{=}\,{\sum}_{ab\lambda}\,\widetilde{P}_{a|x\lambda}\widetilde{P}_{b|y\lambda}\widetilde{P}_{\lambda}\widetilde{P}_{x|\lambda}\widetilde{P}_{y}\,=\,
{\sum}_{\lambda}\,\widetilde{P}_{x|\lambda}\widetilde{P}_{\lambda}\widetilde{P}_{y}\,.
\end{eqnarray}}

\section{Causal explainability of experimental statistics}

Having identified causal strategies for violating the respective assumptions in Bell's theorem, it is natural to ask whether those mechanisms are actually enough to simulate the observed experimental statistics.

Let us formalise the concept of explainability of a given statistics by a given causal diagram.
\begin{definition}[Explainability]
We will say that the experimental statistics $\big\{P_{ab|xy}\,,P_{xy}\big\}$ is explainable (or compatible) with the causal structure \textcolor{blue}{(${*}$)}, \textcolor{blue}{(\textit{NL})}, \textcolor{blue}{(\textit{R})} or \textcolor{blue}{(\textit{NF})} if there exist a choice of parameters $\widetilde{P}$ such that the causal diagram in question reproduces the observed statistics, i.e., respectively Eqs.\,(\ref{base-1})-(\ref{base-2}), Eqs.\,(\ref{NL-1})-(\ref{NL-2}), Eqs.\,(\ref{R-1})-(\ref{R-2}) or Eqs.\,(\ref{NF-1})-(\ref{NF-2}) hold.
\end{definition}

Clearly, the baseline structure \textcolor{blue}{(${*}$)} cannot explain quantum mechanical statistics, since they violate Bell inequalities. For the structures \textcolor{blue}{(\textit{NL})}, \textcolor{blue}{(\textit{R})} or \textcolor{blue}{(\textit{NF})}, Bell inequalities are not an obstacle. 

We also note that the structures \textcolor{blue}{(\textit{NL})}, \textcolor{blue}{(\textit{R})} and \textcolor{blue}{(\textit{NF})}, shown diagrammatically in Figs.~\ref{Fig-Bell-NL}, \ref{Fig-Bell-R} and \ref{Fig-Bell-NF}, cannot explain arbitrary experimental statistics $\big\{P_{ab|xy}\,,P_{xy}\big\}$. From the \textit{d-separation} rules, some conditional independencies must hold. One immediately observes that in each of those diagrams, node $b$ is a collider which blocks all paths between nodes $a$ and $y$. It also blocks paths between $x$ and $y$. Therefore, we will henceforth impose the following assumption.
\begin{assumption}\label{assumption}
Let in the observed statistics $\big\{P_{ab|xy}\,,P_{xy}\big\}$ the following independence conditions hold
\begin{eqnarray}\label{assump}
a\indep y\,|\,x&\quad\&\quad&x\indep y\,.
\end{eqnarray}
\end{assumption}
\noindent This is a \textit{necessary} condition which has to be satisfied by any statistics compatible with \textcolor{blue}{(\textit{NL})}, \textcolor{blue}{(\textit{R})} or \textcolor{blue}{(\textit{NF})}. Note that this asymmetric condition boils down to non-signalling from Bob to Alice. Let us remark that in quantum theory this constraint is satisfied, since non-signalling holds in both directions$^\text{\ref{non-sig}}$. This renders the assumption innocuous for our purposes. Most importantly, it turns out that this is also a \textit{sufficient} condition, as explained in the following lemma (see \textbf{Appendix} for the proof).
\begin{lemma}\label{lemma}
If the experimental statistics $\big\{P_{ab|xy}\,,P_{xy}\big\}$ satisfies \textbf{Assumption~\ref{assumption}} then it is explainable (or compatible) with either of the causal structures \textcolor{blue}{(\textit{NL})}, \textcolor{blue}{(\textit{R})} or \textcolor{blue}{(\textit{NF})}.
\end{lemma}

It means that any of the discussed causal mechanisms can properly account for the observed quantum mechanical correlations. In the following section, we compare those three mechanisms asking about their performance in simulating a given experiment based on an intuitive notion of a fractional measure. 

\section{Causal fraction as the cost of violation of a given assumption}

Suppose that the observed experimental statistics obeys \textbf{Assumption~\ref{assumption}} (for all distributions of settings $P_{xy}$), which means that it can be explained by causal models \textcolor{blue}{(\textit{NL})}, \textcolor{blue}{(\textit{R})} or \textcolor{blue}{(\textit{NF})} according to \textbf{Lemma~\ref{lemma}}. This implies that the additional arrow provides a means to simulate the behaviour $P_{ab|xy}$ for any distribution of settings $P_{xy}$. However, even if deemed required to account for the observations, it may be the case that this additional arrow is excessive to some extent. Perhaps it is enough for the extra arrow to be effective only occasionally to explain the observations. In other words, in the repeated experiment, one may refrain from using it on some fraction of trials and still be able to account for the observed statistics. So to say, the modeller can perform better by being frugal with extra resources/arrows and still achieve the goal of simulating the experiment. 

This idea can be formalised by decomposing the observed experimental statistics $\big\{P_{ab|xy}\,,P_{xy}\big\}$ in the form of a convex combination
\begin{eqnarray}\label{decomp-1}
P_{ab|xy}&=&(1-q
)\,P^{\,*}_{ab|xy}\,+\,q\,P^{\,\#}_{ab|xy}\,,\\\label{decomp-2}
P_{xy}&=&(1-q
)\,P^{\,*}_{xy}\,+\,q\,P^{\,\#}_{xy}\,,
\end{eqnarray}
where $\big\{P^{\,*}_{ab|xy}\,,P^{\,*}_{xy}\big\}$ is a baseline statistics compatible with \textcolor{blue}{($*$)}, and $\big\{P^{\,\#}_{ab|xy}\,,P^{\,\#}_{xy}\big\}$ is a statistics requiring the additional arrow. Here, the label $\#$ is meant as a placeholder for the case in question, i.e., whether it is the extra arrow $x\rightarrow b$ in \textcolor{blue}{(\textit{NL})}, $x\rightarrow \lambda$ in \textcolor{blue}{(\textit{R})} or $\lambda \rightarrow x$ in \textcolor{blue}{(\textit{NF})}.

Eqs.\,(\ref{decomp-1})-(\ref{decomp-2}) have a simple interpretation: \textit{the observed experimental statistics $\big\{P_{ab|xy}\,,P_{xy}\big\}$ can be explained using the additional causal mechanism (red arrow) only in the fraction of $q$ experimental runs, while refraining from its use in the remaining fraction of $1-q$ runs.}

Note that, as it stands above, the statistics  $\big\{P^{\,\#}_{ab|xy}\,,P^{\,\#}_{xy}\big\}$ can still be excessive and hence subject to optimisation. There can be many other competitive decompositions as well. Therefore, in order to get a unique result, one needs to \textit{minimise} $q$ over all possible convex decompositions in Eqs.\,(\ref{decomp-1})-(\ref{decomp-2}). However, we note that such an optimal solution can be sensitive to the choice of the distribution of settings $P_{xy}$. It is not hard to imagine that in experiments with non-trivially distributed settings, one can outperform other choices by exploiting the unevenness of the distribution $P_{xy}$ (e.g., by ignoring rare occurrences). For a fair comparison, one needs to take the worst-case scenario in the simulation task, since it is not a priori known which distribution of settings $P_{xy}$ comes along with the given behaviour $P_{ab|xy}$. This requires \textit{maximising} over all possible distributions of settings $P_{xy}$. As a result, we are left with the figure of merit that is related only to the behaviour $P_{ab|xy}$ (and valid for any distribution of settings $P_{xy}$).

We are led to the following definition of the \textit{fractional causal measure} of the violation of a given assumption in the Bell scenario. 
\begin{definition}[Causal fraction]For a given behaviour $P_{ab|xy}$ in the Bell scenario, we define the fractional causal measure in the following way
\begin{eqnarray}\label{causal-fraction}
\mu^{\scriptscriptstyle \#}&:=&\max_{\scriptscriptstyle P_{xy}}\ \min_{decomp.\atop (\ref{decomp-1})\text{-}(\ref{decomp-2})}\ q\,,
\end{eqnarray}
where $\#$ denotes the case under consideration, i.e., whether it is non-locality with the arrow $x\rightarrow b$ in \textcolor{blue}{(\textit{NL})}, retrocausality with the arrow $x\rightarrow \lambda$ in \textcolor{blue}{(\textit{R})} or the arrow $\lambda \rightarrow x$ violating free choice in \textcolor{blue}{(\textit{NF})}. This defines three measures $\mu^{\scriptscriptstyle NL}$, $\mu^{\scriptscriptstyle R}$ and $\mu^{\scriptscriptstyle NF}$.
\end{definition}
\noindent The causal fraction is designed so to assess the strength of a given causal arrow by answering a simple question:
\textit{"How often a given red arrow needs to be in effect to reproduce the observed behaviour $P_{ab|xy}$ for every possible distribution of settings $P_{xy}$?"} The expression "how often" is understood as the least fraction of trials in a repeated experiment when the red arrow has to be actually used to perform the simulation task.

The concept of causal fraction has broad applicability. As defined in Eq.\,(\ref{causal-fraction}), it is tailored to \textit{individually} assess the violation of a given assumption in the Bell scenario, i.e. whether the red arrow in question is $x\rightarrow b$ in \textcolor{blue}{(\textit{NL})}, $x\rightarrow \lambda$ in \textcolor{blue}{(\textit{R})}  or $\lambda \rightarrow x$ in \textcolor{blue}{(\textit{NF})}. Most importantly, this allows for a fair comparison of the causal assumptions in the Bell scenario based on the measure built on the very same notion of the fractional measure.

It may be a surprise that despite conceptual differences (and interpretational implications) between those assumptions, there is no difference in their fractional causal cost. This is shown by the following general structural result about the weight of causal explanations in the Bell scenario.

\newpage

\begin{theorem}\label{Theorem}
For any behaviour $P_{ab|xy}$ satisfying \textbf{Assumption~\ref{assumption}} (for all distributions $P_{xy}$), 
the fractional causal measures in Eq.\,(\ref{causal-fraction}) defined for non-local, retrocausal and non-free explanations are all the \underline{same}, i.e. we have
\begin{eqnarray}
\mu^{\scriptscriptstyle NL}\ =\ \mu^{\scriptscriptstyle R}\ =\ \mu^{\scriptscriptstyle NF}\,,
\end{eqnarray}
where the labels NL, R and NF stand respectively for the non-local arrow $x\rightarrow b$ in \textcolor{blue}{(\textit{NL})}, retrocausal arrow $x\rightarrow \lambda$ in \textcolor{blue}{(\textit{R})}, and non-free arrow $\lambda \rightarrow x$ in \textcolor{blue}{(\textit{NF})}.
\end{theorem}
\begin{proof}
For the sake of generality, let the label \# be a placeholder for whichever case \textcolor{blue}{(\textit{NL})}, \textcolor{blue}{(\textit{R})} or \textcolor{blue}{(\textit{NF})}. Let the behaviour $P_{ab|xy}$ be fixed and choose some distribution of settings $P_{xy}$ which admit decomposition in Eqs.~(\ref{decomp-1})-(\ref{decomp-2}), i.e.,
\begin{eqnarray}
&&\big\{P^{\,*}_{ab|xy}\,,P^{\,*}_{xy}\big\}\quad\text{is a baseline statistics \textcolor{blue}{($*$)}, and}\\
&&\big\{P^{\,\#}_{ab|xy}\,,P^{\,\#}_{xy}\big\}\quad\text{is a statistics compatible with \#\,.}
\end{eqnarray}
Note that the latter statistics being compatible with \# has to satisfy \textbf{Assumption~\ref{assumption}}, which entails by \textbf{Lemma~\ref{lemma}} that it is also compatible with any other $\#^{\,\prime}$. Thus we can also write
\begin{eqnarray}\label{decomp-1'}
P_{ab|xy}&=&(1-q)\,P^{\,*}_{ab|xy}\,+\,q\,P^{\,\#^{\,\prime}}_{ab|xy}\,,\\\label{decomp-2'}
P_{xy}&=&(1-q)\,P^{\,*}_{xy}\,+\,q\,P^{\,\#^{\,\prime}}_{xy}\,,
\end{eqnarray}
where $\big\{P^{\,\#^{\,\prime}}_{ab|xy}\,,P^{\,\#^{\,\prime}}_{xy}\big\}$ is compatible with $\#^{\,\prime}$. This establishes for the statistics $\big\{P_{ab|xy}\,,P_{xy}\big\}$ a one-to-one correspondence between the decompositions in Eqs.\,(\ref{decomp-1})-(\ref{decomp-2}) and Eqs.\,(\ref{decomp-1'})-(\ref{decomp-2'}), each pertaining to the respective case $\#$ and $\#^{\,\prime}$. Since the weight $q$ is the same, we infer that 
\begin{eqnarray}
\min_{decomp.\atop (\ref{decomp-1})\text{-}(\ref{decomp-2})}\ q&=&\min_{decomp.\atop (\ref{decomp-1'})\text{-}(\ref{decomp-2'})}\ q\,.
\end{eqnarray}
Since it holds for every $P_{xy}$, we conclude that $\mu^{\scriptscriptstyle \#}=\mu^{\scriptscriptstyle \#^{\,\prime}}$. This ends the proof. 
\end{proof}

\section{Discussion}

The research in the paper concerns the analysis of assumptions in Bell's theorem. The adopted framework is based on the causal Bayesian models by Pearl and others~\cite{Pe09,SpGlSc00,PeGlJe16}, which subsumes Bell's original ideas~\cite{Be93,Va82}. It has the advantage of articulating the underlying assumptions and their possible violations in causal language with a simple graphical representation. Our main focus in the paper is the question of whether the formalism in itself can adjudicate between different causal mechanisms/arrows leading to the violation of the respective assumptions, without bias regarding their philosophical/interpretative consequences. We show that the cost of relaxing either assumption --- \textit{locality, arrow-of-time, or free choice} --- is always the \textit{same}, see \textbf{Theorem~\ref{Theorem}}. This result is formally established via the so-called \textit{causal fraction measure}, which computes the least frequency of violation of a given assumption required to simulate the observed behaviour. In each case, the measure is defined in a similar manner in order to provide a unified basis for a fair comparison. This completes the results in~\cite{BlPoYeGaBo21} with the discussion of retrocausal models.

In this work we insist on a single-arrow type of violation of a given assumption, see \textcolor{blue}{(\textit{NL})}, \textcolor{blue}{(\textit{R})} or \textcolor{blue}{(\textit{NF})}. That is, we consider the most modest scenarios that are enough to explain the observed statistics. We are guided here by the paradigm in graphical causal models, saying that the fewer arrows, the less expressive the power of the model (or, in other words, if a given causal graph is able to explain the statistics, then a model with more arrows will be able to explain it too). Therefore, the fewer arrows (or simpler DAG), the better, as far as causal explanations go. From this point of view, adding only a single extra arrow to the baseline scenario \textcolor{blue}{(${*}$)} makes explaining the desired statistics more challenging and interesting at the same time. This is what we were able to achieve in this paper (cf. \textbf{Lemma~\ref{lemma}}).

The results presented here are derived under \textbf{Assumption~\ref{assumption}} which is the consequence of single-arrow type violations in \textcolor{blue}{(\textit{NL})}, \textcolor{blue}{(\textit{R})} or \textcolor{blue}{(\textit{NF})}. Thus, we admit a broader range of experimental statistics with one-sided signalling only. Let us remark that even in the case of non-signalling correlations in both directions (like in quantum mechanics) this may add interpretational flexibility. This might be useful, for example, when the experimental design needs the asymmetric causal structure explaining the results (e.g., due to the specific temporal arrangement of events).

Note that we take a viewpoint where causal concepts are treated as fundamental in physical theories. Furthermore, we adopt the approach to causality as laid out by Pearl and others~\cite{Pe09,SpGlSc00,PeGlJe16}. Of course, one may dismiss the importance of causality as such, seeing it as an illusory or emergent phenomenon~\cite{No03}. Another approach is considering the possibility of a different causal framework on the fundamental level. See~\cite{CaLa14a,WiCa17,AlBaHoLeSp17,Dz22a,CaLe23,Ad23} for a few research directions along those lines. See also a discussion of the relationship between context independence and freedom of choice in~\cite{Dz22,Dz22a,Dz23}.

We remark that the models constructed in this work are necessarily fine-tuned (cf. \textbf{Lemma~\ref{lemma}}). It is in line with the result in~\cite{WoSp15}, where it was shown that the violation of Bell inequalities together with non-signalling conditions entails fine-tuning. Furthermore, we are not concerned with the overhead costs of simulation, i.e. the cardinality of the hidden variables is unrestricted. Although this is a standard assumption in the causal inference framework, it does not have to be the case~\cite{ZjWoSp21}.

Finally, we note that the causal fraction has been considered in the literature to measure non-locality in the Bell scenario and explicitly calculated in a few cases~\cite{Ha91,ElPoRo92,BaKePi06,CoRe08,PoBrGi12,AbBr11,AbBaMa17}. See\cite{BlPoYeGaBo21} for a compilation of those results. As for the arrow-of-time and free choice assumption, all quantitative measures designed in the literature are based on completely different notions. See~\cite{HaBr20} for retrocausality and~\cite{Ha10,Ha11,Ha16,BaGi11,PuRoBaLiGi14} for violation of free choice. Therefore, comparing those measures does not serve the purpose of adjudicating on the cost of violation of the respective assumptions in the Bell experiment. The present paper gives an impartial treatment of all three causal assumptions underlying Bell's theorem based on the unifying concept of causal fractions. 

\vskip10pt

\ack{\ \\PB acknowledges support from the Polish-U.S. Fulbright Commission and the hospitality of the Institute of Quantum Studies at Chapman University. We thank Matt Leifer, Emmanuel Pothos and James Yearsley for inspiring discussions on the subject.}

\vskip30pt

\section*{Appendix: Proof of \textit{\textbf{Lemma~\ref{lemma}}}}

From the \textbf{Assumption~\ref{assumption}} we have that 
\begin{eqnarray}\label{Pa|x}
P_{a|x}\ \stackrel{\scriptscriptstyle(\ref{assump})}{=}\ P_{a|xy}\ =\ \sum_b\, P_{ab|xy}
\end{eqnarray}
is well-defined given the behaviour $P_{ab|xy}$. Furthermore, the distribution of settings factorises
\begin{eqnarray}\label{Pxy}
P_{xy}\ \stackrel{\scriptscriptstyle(\ref{assump})}{=}\ P_{x}\,P_{y}\,,
\end{eqnarray}
where $P_x={\sum}_y\,P_{xy}$ and $P_y={\sum}_x\,P_{xy}$ are marginals of $P_{xy}$.

It will be useful to observe that the distribution $P_{a|x}$ in Eq.\,(\ref{Pa|x}) can be decomposed as a convex combination of deterministic strategies $\delta_{a=f(x)}$. That is, we have
\begin{eqnarray}\label{Pa|x-det}
P_{a|x}&=&\sum_{f\in\mathfrak{F}}\,\delta_{a=f(x)}\,P_f\,,
\end{eqnarray}
where $\mathfrak{F}:=\big\{\,f:\mathcal{X}\rightarrow \mathcal{A}\,\big\}\equiv\mathcal{A}^\mathcal{X}$ is the set of functions from $\mathcal{X}$ to $\mathcal{A}$, and $P_f$ is a well-normalized distribution ${\sum}_f\,P_{f}=1$.  

For the sake of simplicity, let us assume that all probabilities are non-zero, so as to comfortably use  Bayes' rule without bothering with divisions by zero (those cases can be treated separately in a straightforward manner).

For the proof, in each particular causal model \textcolor{blue}{(\textit{NL})}, \textcolor{blue}{(\textit{R})} and \textcolor{blue}{(\textit{NF})}, we will need to explicitly define the hidden variable space $\Lambda$, give the corresponding set of parameters $\widetilde{P}$, and show that it reproduces the desired experimental statistics $\big\{P_{ab|xy}\,,P_{xy}\big\}$ via the respective Eqs.\,(\ref{NL-1})-(\ref{NL-2}), Eqs.\,(\ref{R-1})-(\ref{R-2}) and Eqs.\,(\ref{NF-1})-(\ref{NF-2}).

\begin{proof}[Proof of case \textcolor{blue}{(\textit{NL})}]\ \\
Let the hidden variable space be as follows
\begin{eqnarray}
&&\Lambda\ :=\ \mathfrak{F}\,,\quad\text{i.e. we have}\quad \lambda\ \equiv\ f\,,
\end{eqnarray}
 and define the parameters to be
\begin{eqnarray}
\widetilde{P}_{a|x\lambda}&:=&\delta_{a=f(x)}\,,\\
\widetilde{P}_{b|xy\lambda}&:=&\frac{P_{f(x)b|xy}}{P_{f(x)|x}}\,,\\
\widetilde{P}_{\lambda}&:=&{P}_f\,,\\
\widetilde{P}_x&:=&P_x\,,\\
\widetilde{P}_y&:=&P_y\,.
\end{eqnarray}
Then we calculate
\begin{eqnarray}
{\sum}_\lambda\,\widetilde{P}_{a|x\lambda}\,\widetilde{P}_{b|xy\lambda}\,\widetilde{P}_{\lambda}&=&
{\sum}_f\ \delta_{a=f(x)}\ \frac{P_{f(x)b|xy}}{P_{f(x)|x}}\ P_{f}\\
&=&{\sum}_f\ \delta_{a=f(x)}\ \frac{P_{ab|xy}}{P_{a|x}}\ \ P_{f}\ \stackrel{\scriptscriptstyle(\ref{Pa|x-det})}{=}\ \ \frac{P_{ab|xy}}{P_{a|x}}\ P_{a|x}\ \ =\ \ P_{ab|xy}\,,
\end{eqnarray}
which proves Eq.\,(\ref{NL-1}), and
\begin{eqnarray}
\widetilde{P}_x\,\widetilde{P}_y&=&P_x\,P_y\ \ \stackrel{\scriptscriptstyle(\ref{Pxy})}{=}\ \ P_{xy}\,,
\end{eqnarray}
which proves Eq.\,(\ref{NL-2}).
\end{proof}

\begin{proof}[Proof of case \textcolor{blue}{(\textit{R})}]\ \\
Here, we take the hidden variable space in the form
\begin{eqnarray}
&&\Lambda\ :=\ \mathfrak{F}\times \mathcal{X}\,,\quad\text{i.e. we have}\quad \lambda\ \equiv\ (f,\tilde{x})\,,
\end{eqnarray}
and define the parameters in the following way
\begin{eqnarray}
\widetilde{P}_{a|x\lambda}&:=&\delta_{a=f(x)}\,,\\
\widetilde{P}_{b|y\lambda}&:=&\frac{P_{f(\tilde{x})b|\tilde{x}y}}{P_{f(\tilde{x})|\tilde{x}}}\,,\\
\widetilde{P}_{\lambda|x}&:=&{P}_f\ \delta_
{x=\tilde{x}}\,,\\
\widetilde{P}_x&:=&P_x\,,\\
\widetilde{P}_y&:=&P_y\,.
\end{eqnarray}
This gives 
\begin{eqnarray}
\sum_\lambda\,\widetilde{P}_{a|x\lambda}\,\widetilde{P}_{b|y\lambda}\,\widetilde{P}_{\lambda|x}&=&
\sum_{f,\tilde{x}}\ \delta_{a=f(x)}\ \frac{P_{f(\tilde{x})b|\tilde{x}y}}{P_{f(\tilde{x})|\tilde{x}}}\ P_{f}\ \delta_
{x=\tilde{x}}\\
&=&
\sum_{f}\ \delta_{a=f(x)}\ \frac{P_{f({x})b|{x}y}}{P_{f({x})|{x}}}\ P_{f}\\
&=&
\sum_{f}\ \delta_{a=f(x)}\ \frac{P_{ab|{x}y}}{P_{a|{x}}}\ P_{f}
\ \stackrel{\scriptscriptstyle(\ref{Pa|x-det})}{=}\ \ \frac{P_{ab|xy}}{P_{a|x}}\ P_{a|x}\ \ =\ \ P_{ab|xy}\,,
\end{eqnarray}
which proves Eq.\,(\ref{R-1}), and
\begin{eqnarray}
\widetilde{P}_x\,\widetilde{P}_y&=&P_x\,P_y\ \ \stackrel{\scriptscriptstyle(\ref{Pxy})}{=}\ \ P_{xy}\,,
\end{eqnarray}
which proves Eq.\,(\ref{R-2}).
\end{proof}

\begin{proof}[Proof of case \textcolor{blue}{(\textit{NF})}]\ \\
In this case, let the hidden variable space be as follows
\begin{eqnarray}
&&\Lambda\ :=\ \mathcal{A}\times \mathcal{X}\,,\quad\text{i.e. we have}\quad \lambda\ \equiv\ (\tilde{a},\tilde{x})\,,
\end{eqnarray}
and put the parameters in the following form
\begin{eqnarray}
\widetilde{P}_{a|x\lambda}&:=&\delta_{a=\tilde{a}}\,,\\
\widetilde{P}_{b|y\lambda}&:=&\frac{P_{\tilde{a}b|\tilde{x}y}}{P_{\tilde{a}|\tilde{x}}}\,,\\
\widetilde{P}_{\lambda}&:=&{P}_{\tilde{a}\tilde{x}}\,,\\
\widetilde{P}_{x|\lambda}&:=&\delta_{x=\tilde{x}}\,,\\
\widetilde{P}_y&:=&P_y\,.
\end{eqnarray}
Note that the distribution $P_{ax}$ obtains either via $P_{ax}\ =\ P_{a|x}\,P_x$ or $P_{ax}\ =\ {\sum}_{b,y}\,P_{ab|xy}\,P_{xy}$\,.

\noindent From those definitions by the Bayes' rule, it follows that
\begin{eqnarray}\label{Plambda|x}
\widetilde{P}_{\lambda|x}&=&\frac{\widetilde{P}_{x|\lambda}\ \widetilde{P}_{\lambda}}{\widetilde{P}_{x}}\ \ =\ \ \frac{\delta_{x=\tilde{x}}\ P_{\tilde{a}\tilde{x}}}{P_x}\ \ =\ \ \frac{\delta_{x=\tilde{x}}\ P_{\tilde{a}\tilde{x}}}{P_{\tilde{x}}}\ \ =\ \ \delta_{x=\tilde{x}}\ P_{\tilde{a}|\tilde{x}}\,,
\end{eqnarray}
where in the denominator we have used the fact that
\begin{eqnarray}
\widetilde{P}_x\ =\ {\sum}_\lambda\,\widetilde{P}_{x|\lambda}\,\widetilde{P}_\lambda\ =\ \sum_{\tilde{a},\tilde{x}}\,\delta_{x=\tilde{x}}\,P_{\tilde{a}\tilde{x}}\ =\ \sum_{\tilde{a}}\,P_{\tilde{a}x}\ =\ P_{x}\,.
\end{eqnarray}
Then we have 
\begin{eqnarray}
\sum_\lambda\,\widetilde{P}_{a|x\lambda}\,\widetilde{P}_{b|y\lambda}\,\widetilde{P}_{\lambda|x}&\stackrel{\scriptscriptstyle(\ref{Plambda|x})}{=}&
\sum_{\tilde{a},\tilde{x}}\ \delta_{a=\tilde{a}}\ \frac{P_{\tilde{a}b|\tilde{x}y}}{P_{\tilde{a}|\tilde{x}}}\ \delta_{x=\tilde{x}}\ P_{\tilde{a}|\tilde{x}}\ \ =\ \ P_{ab|xy}\,,
\end{eqnarray}
which proves Eq.\,(\ref{NF-1}), and
\begin{eqnarray}
\sum_\lambda\,\widetilde{P}_{x|\lambda}\,\widetilde{P}_\lambda\,\widetilde{P}_y&=&\sum_{\tilde{a},\tilde{x}}\,P_{\tilde{a}\tilde{x}}\,\delta_{x=\tilde{x}}\,{P}_y\ \ =\ \ \sum_{\tilde{a}}\,P_{\tilde{a}{x}}\,{P}_y\ \ =\ \ {P}_x\,{P}_y\ \ \stackrel{\scriptscriptstyle(\ref{Pxy})}{=}\ \ P_{xy}\,,
\end{eqnarray}
which proves Eq.\,(\ref{NF-2}).
\end{proof}

\newpage


\bibliographystyle{RS}
\bibliography{CombQuant}

\end{document}